\documentclass{article}

\usepackage[normalem]{ulem}

\usepackage{amsmath}
\usepackage{amsfonts}
\usepackage{amssymb}	
\usepackage{color}
\usepackage[dvipsnames]{xcolor}                                         

\usepackage[colorlinks,citecolor=blue,urlcolor=blue]{hyperref}          

\newtheorem{thm}{Theorem}[section]
\newtheorem{lem}[thm]{Lemma}
\newtheorem{cor}[thm]{Corollary}
\newtheorem{rem}[thm]{Remark}
\newtheorem{prop}[thm]{Proposition}

\newcommand{\mysection}[1]{\section{#1} \setcounter{equation}{0}}
\newcommand{\proof}{{\sc Proof.} \quad}

\newcommand{\R}{\mathbb{R}}

\newcommand{\be}{\begin{equation} \label}
\newcommand{\ee}{\end{equation}}
\newcommand{\bes}{\begin{equation} \begin{array}{c} \label}
\newcommand{\ees}{\end{array} \end{equation}}
\newcommand{\bea}{\begin{eqnarray}\label}
\newcommand{\eea}{\end{eqnarray}}
\newcommand{\beas}{\begin{eqnarray} \begin{array}{rcl} \label}
\newcommand{\eeas}{\end{array} \end{eqnarray}}
\newcommand{\bas}{\begin{eqnarray*}}\newcommand{\eas}{\end{eqnarray*}}
\newcommand{\bass}{\begin{eqnarray*} \begin{array}{rcl}}
\newcommand{\eass}{\end{array} \end{eqnarray*}}
\newcommand{\basss}{\begin{eqnarray*} \begin{array}{c}}
\newcommand{\easss}{\end{array} \end{eqnarray*}}

\newcommand{\bit}{\begin{itemize}}
\newcommand{\eit}{\end{itemize}}

\newcommand{\abs}{\\[3mm]}

\newcommand{\w}{\omega}

\newcommand{\CC}{{\mathbb{C}}}

\newcommand{\set}{{\cal S}(m,M,B)}

\newcommand{\8}{\infty}%
\renewcommand{\abs}[1]{\left| #1 \right|}
\renewcommand{\set}[1]{\left\{ #1 \right\}}
\newcommand{\norm}[1]{\left\| #1 \right\|}

\newcommand{\ov}[1]{\overline{ #1 }}

\newcommand{\cv}{\mathfrak{c}}
\renewcommand{\a}{\alpha}
\newcommand{\e}{\varepsilon}

\begin{document}

\title{Kinetic energy represented in terms of moments of vorticity and applications.}
\author{
Tomasz Cie\'slak \\
{\small Institute of Mathematics, Polish Academy of Sciences, \'Sniadeckich 8, Warsaw, Poland}\\
{\small e-mail: T.Cieslak@impan.pl}\\
\\
Krzysztof Oleszkiewicz\\
{\small Institute of Mathematics, Polish Academy of Sciences, \'Sniadeckich 8, Warsaw, Poland}\\
{\small Institute of Mathematics, University of Warsaw, ul. Banacha 2, Warsaw, Poland}\\
{\small e-mail: koles@mimuw.edu.pl}\\
\\
Marcin Preisner\\
{\small Instytut Matematyczny, Uniwersytet Wroc\l{}awski, Pl. Grunwaldzki 2/4, Wroc\l{}aw, Poland}\\
{\small e-mail: marcin.preisner@uwr.edu.pl}\\
\\
Marta Szuma\'nska\\
{\small Institute of Mathematics, University of Warsaw, ul. Banacha 2, Warsaw, Poland}\\
{\small Institute of Mathematics, Polish Academy of Sciences, \'Sniadeckich 8, Warsaw, Poland}\\
{\small e-mail: M.Szumanska@mimuw.edu.pl}}

\maketitle

\begin{abstract} We study 2d vortex sheets with unbounded support. First we show a version of the Biot-Savart law related to a class of objects including such vortex sheets. Next, we give a formula associating the kinetic energy of a very general class of flows with certain moments of their vorticities. It allows us to identify a class of vortex sheets of unbounded support being only $\sigma$-finite measures (in particluar including measures $\omega$ such that $\omega(\R^2)=\infty$), but with locally finite kinetic energy. One of such examples are celebrated Kaden approximations. We study them in details. In particular our estimates allow us to show that the kinetic energy of Kaden approximations in the neighbourhood of an origin is dissipated, actually we show that the energy is pushed out of any ball centered in the origin of the Kaden spiral. The latter result can be interpreted as an artificial viscosity in the center of a spiral.

\end{abstract}

\mysection{Introduction}\label{intro}

In the present paper we study vortex sheets which are not compactly supported.
In engineering or physics literature vortex sheets are usually two-dimensional divergence-free (in the sense of distributions) vector
fields such that their vorticities are zero except on a curve $\cv$, along which tangential components of
velocity are discontinuous. For us vortex sheets are a wider class of objects, namely 2d divergence-free velocity fields
whose vorticity $\omega$ are $\sigma-$ finite measures only. In particular unbounded measures $\omega$, i.e. such that
\[
\omega(\R^2)=\infty,
\]
are included. The usual definition of vortex sheets
assumes that such objects have vorticities being compactly supported finite Radon measures, see
\cite{DiP_Majda}. However such restriction eliminates from the considerations well-known
spirals of vorticity, self-similar objects well-established in engineering and physics literature
like Kaden spirals (see \cite{kaden, elling}), Prandtl spirals (we refer the reader to \cite{prandtl,
kambe, saffman, TCMS}) or recent hyperbolic spirals introduced in \cite{sohn}. Extension of the theory to such self-similar vortex spirals seems important and required. Moreover we would like to restrict ourselves to vector fields with locally finite kinetic energy. The importance
of such objects is emphasized in the introduction of \cite{TCMS}. It was noticed in \cite{TCMS} that a crucial property of a compactly supported vorticity measure $\omega$
\begin{equation}\label{wazne}
\omega(B(0,r))=c r^{\alpha},
\end{equation}
where $c$ and $\alpha$ are positive constants, yields that local $H^{-1}$-norm of $\omega$ is finite.
By the lemma of Schochet \cite{ss} it means that the kinetic energy generated by a compactly supported
part of such a vortex sheet is locally finite.
Property \eqref{wazne} is satisfied at least
by well-known examples of Kaden and Prandtl, see \cite{TCMS}.

The main concern of the present paper is to extend the previous study in \cite{TCMS} to the case of vortex sheets which are not necessarily compactly supported, moreover such that their vorticity $\omega$ is an unbounded measure, i.e.
\[
\omega(\R^2)=\infty.
\]One of the main questions we address is whether the kinetic energy generated by such objects is locally finite or not. We shall give precise conditions yielding sharp estimates of local kinetic energy from above and below for a class of objects satisfying \eqref{wazne}, see Theorem \ref{mainthm1} below.

When speaking about kinetic energy carried by a vorticity we need to know the divergence-free velocity associated
to the vorticity. In case of vorticity being a compactly supported regular function, velocity is given by the usual Biot-Savart operator. One of the tasks of the present paper is to identify the velocity given by a vorticity being $\sigma$-finite measure satisfying \eqref{wazne}. This is discussed in Section \ref{vv}. We shall say more on it also at the end of Introduction.

The question concerning kinetic energy is very important for several reasons. On the one hand it is a natural expectation for an object which is supposed to have a physical meaning.
Next, when looking for vortex sheets weak solutions of the 2d Euler equations one has to make sure that the following integral
\[
\int_0^T\int_{\R^2} v(x,t)\otimes v(x,t):\nabla \phi(x,t) dxdt
\]
is finite for a divergence-free velocity field $v:\R^2\rightarrow \R^2$ and any smooth compactly supported
divergence-free test function $\phi:\R^2\rightarrow \R^2$. To this end it suffices that  the local kinetic energy of $v$, which is defined by
    \begin{equation}\label{def_energy}
    E_r(\w) = \int_{B(0,r)} |v(x)|^2 \, dx,
    \end{equation}
is finite.

Our main theorem states that the local kinetic energy $E_r(\w)$ of a nonnegative $\sigma$-finite measure of vorticity $\w$ satisfying \eqref{wazne} undergoes a precise estimate from below and above.
\begin{thm}\label{mainthm1}
Let $\w$ be a nonnegative $\sigma$-finite measure satisfying \eqref{wazne} with $\a\in (0,1)$. Then, for $c$ being a constant appearing in \eqref{wazne}, we have
\begin{equation}\label{mainthm1est}
    \frac{c^2}{4\pi\a} \, r^{2\a} \leq  E_r(\w) \leq \frac{c^2\a\pi}{4\sin^2(\pi \a)}  \, r^{2\a}, \qquad r>0 .
\end{equation}
\end{thm}
In the proof of Theorem \ref{mainthm1} we study the spherical averages
\begin{equation}
    A_r(\w) = (2\pi)^{-1} \int_0^{2\pi} \abs{v(re^{i\theta})}^2 \, d\theta, \quad r>0
\end{equation}
related to $E_r(\w)$ by
    \begin{equation}\label{radial}
    E_r(\w) = 2\pi \int_0^r A_s(\w)\, s\, ds.
    \end{equation}
Our main tool to estimate $A_r(\w)$ is the following formula expressing $A_r(\w)$ in terms of inner and outer moments of vorticity. Since 2d plane can be viewed as a set of complex numbers $z=x+iy$, for $r>0$ and $n\geq 1$ let us denote
\begin{align}
    \label{moment0}
    {m}_{r,0}(\w) &= \w(B(0,r)),\\
    \label{moment1}
    {m}_{r,n}(\w) &= \int_{B(0,r)} u^n\, d\w(u),\\
    \label{moment2}
    {M}_{r,n}(\w) &= \int_{\CC\setminus B(0,r)} u^{-n}\, d\w(u),
\end{align}
where the powers are taken in the sense of complex numbers.

Actually, our formula for $A_r(\w)$ is proven under less restrictive assumptions on the measure $\w$ than \eqref{wazne}. Assume that a nonnegative $\sigma$-finite measure $\w$ on $\R^2$ is given such that
\begin{align}
    \label{mu-bdd2}
    \int_{\R^2} (1+|x|)^{-1} d\w(x) < \8.
\end{align}
For such measures the following result holds. 
\begin{thm}\label{averages}
Assume that a nonnegative $\sigma$-finite measure $\w$ satisfies \eqref{mu-bdd2}. Then
\begin{equation}\label{averages_main}
    4\pi^2 A_r(\w) = \sum_{n=0}^\8 r^{-2n-2} \abs{m_{r,n}(\w)}^2 + \sum_{k=1}^\8 r^{2k-2} \abs{M_{r,k}(\w)}^2
\end{equation}
for a.e. $r>0$.
\end{thm}
The above formula is a completely new way of viewing the relation between velocity and vorticity. We utilize it
in the last section.
\begin{rem}
Let us emphasize that Theorems \ref{mainthm1} and \ref{averages} state results holding for measures which are only $\sigma$-finite, in particular
it applies also to measures $\omega$ such that $\omega(\R^2)=\infty$ (provided they satisfy the assumptions).
\end{rem}

Moreover, we study Kaden's approximations, see \cite{kaden}, and examine their properties. In particular time evolution of the energy in any ball surrounding the origin of the Kaden spiral is computed. It is shown that such an energy is dissipated. When time approaches infinity, the kinetic energy contained in any ball surrounding the origin of the Kaden approximation tends to the minimal possible value given by the left-hand side of \eqref{mainthm1est}, while for small times Kaden's spiral's kinetic energy approaches the upper bound in \eqref{mainthm1est}. Actually, we even compute the limiting objects reached by the divergence-free velocities related to the Kaden spiral when $t$ approaches $0$ as well as when $t$ tends to infinity. The results concerning time evolution are obtained using our moment formula \eqref{averages_main} applied to the difference of two unbounded measures. Such a difference does not have to be a signed measure. Thus, extension of \eqref{averages_main} requires a precise definition of some new objects.

First of all, let us notice that when we consider a difference of two nonnegative measures $\omega_1$ and $\omega_2$, such that $\w_1(\mathbb{R}^2) = \w_2(\mathbb{R}^2)=\infty$, then $\tilde{\omega}:=\omega_1-\omega_2$ cannot be defined as a signed measure. Indeed, one has a problem to decide what is the value $\tilde{\omega}(A)$, where $A$ is such that $\omega_i(A)=\infty$ for $i=1,2$. Such technical difficulties are the main reason for which we work with the following objects, which we shall call {\it vorticities}. By definition, a {\it vorticity} is a distribution $\tilde{\omega}$ which can be represented as a difference of $\omega_1$ and $\omega_2$, such that for $i=1,2$ nonnegative $\sigma$-finite measures $\w_i$ satisfy \eqref{mu-bdd2}.

In Section \ref{kinen} we extend \eqref{mainthm1est} to vorticities $\tilde{\omega}=\omega_1-\omega_2$, for which $\w_i$ satisfy \eqref{mu-bdd2}. Complex moments $m_{r,n}$ and $M_{r,n}$, appearing in \eqref{mainthm1est}, as linear in $\w$ are extended to $\tilde{\w}$ in a natural way.

We also show an extension of the Biot-Savart law to  vorticities from a wide subclass of $\sigma$-finite measures, in particular our results are applicable also to measures $\omega$ such that $\omega(\R^2)=\infty$. Actually, again our theorem works for quite a wide class of objects.   
Moreover, for $\sigma$-finite measures satisfying \eqref{wazne} with $\alpha\in(0,1)$ it is shown that velocities obtained via the Biot-Savart law are in $L^2_{loc}$, so that the kinetic energy is finite.

Let us explain ourselves from the slightly non-orthodox structure of the paper. Namely, the technical core with the proof of Theorem \ref{averages} is a content of Subsection \ref{srednie}. It is self-contained. Moreover, some of main results which we prove with the help of Theorem \ref{averages} are presented in earlier sections. We use there results proven later in Subsection \ref{srednie}. This way technical computations are postponed.

At the end of Introduction let us state an easy fact concerning the assumptions we provided for measures $\w$. We show that measures on $\CC$ satisfying \eqref{wazne} with $\a\in (0,1)$ satisfy also the assumption \eqref{mu-bdd2}. Indeed, we have the following proposition.
\begin{prop}\label{sprawdz}
Let us fix $\a \in (0,1)$. Let a nonnegative $\sigma$-finite measure $\w$ satisfy \eqref{wazne} with such $\a$. Then \eqref{mu-bdd2}
is also satisfied.
\end{prop}
\begin{proof}
Notice that
\begin{equation*}
\begin{split}
    \int_\CC \frac{d\w(u)}{1+|u|} &= \int_0^\8 \w\left(\set{\frac{1}{1+|u|}>t}\right) \, dt = \int_0^1 \w\left(B(0, 1/t-1)\right)\, dt\\
    &= c\int_0^1 \left(\frac{1}{t}-1\right)^\a\, dt \leq c\int_0^1 t^{-\a}\, dt <\8,
\end{split}\end{equation*}
and so \eqref{mu-bdd2} holds true.

\hfill $\Box$
\end{proof}

However, the class of measures satisfying \eqref{mu-bdd2} is much wider than
those fulfiling \eqref{wazne}, for instance measures considered in \cite{jamroz} are also fine.

\vspace{0.3cm}
{\bf Notation.} We need to fix a convention which we use to speak about 2d vorticity. By definition, $(x_1,x_2)^\perp = (-x_2, x_1)$. The partial derivatives in $\R^2$ are denoted by $\partial_1$ and $\partial_2$, while a vorticity of a vector field $v=(v_1,v_2)$ is defined as $\mathrm{curl} (v)=\partial_1v_2-\partial_2v_1$. When it is convenient, we shall use complex notation for $\R^2$.

\mysection{Vorticity and velocity. Biot-Savart law}\label{vv}

Main objects of our studies are two-dimensional velocity fields associated with vortex sheets. In the case when a vorticity $\w$ related to the divergence-free velocity $v=(v_1,v_2)$ is a regular compactly supported function, $\mathrm{div} (v)=0$ and so there exists a potential $\psi$ such that $\nabla \psi=v^\perp$. In other words $v=-^\perp\nabla \psi$. Taking the curl of both sides we arrive at
\begin{equation}\label{funkcja_pradu}
-\triangle \psi = \omega.
\end{equation}
Then the velocity field $v$ is recovered from vorticity $\w$ by the Biot-Savart formula, i.e.,
\begin{equation}
\label{Biot_Savart}
v(x) = \frac{1}{2\pi} \int_{\R^2}\frac{(x-y)^\perp}{|x-y|^2}d\omega(y).
\end{equation}
The same procedure works for more general vorticities, even for compactly supported measures. Then
$\psi$, a solution to \eqref{funkcja_pradu}, still exists, and if $\w\in H^{-1}_{loc}$ then $\psi$ is
regular enough to make sure that $v$ is given by \eqref{Biot_Savart}. However it is not known whether the Biot-Savart law is still valid in the case of a vorticity which is not a bounded measure. Yet, a very important class of vortex sheets are the so-called self-similar spirals of vorticity (one of them, the Kaden spiral is studied later on in the present paper). For instance measures satisfying \eqref{wazne} are immediately unbounded, indeed
\[
\omega(\R^2)=\lim_{r\rightarrow \infty} \omega(B(0,r))=\lim_{r\rightarrow \infty} cr^{\alpha}=\infty.
\]

The question which we address is a validity of Biot-Savart's law for nonnegative measures satisfying \eqref{mu-bdd2}. We show that for such objects the Biot-Savart formula is still well-defined and recovers velocity field $v$ related to vortex sheet $\w$, i.e. $\mathrm{curl}\, (v) =\w$, provided $\w$ satisfies \eqref{mu-bdd2}. Consequently, in view of Proposition \ref{sprawdz}, the Biot-Savart law holds in particular for $\sigma$-finite measures satisfying $\eqref{wazne}$ with $\a\in(0,1)$. We show that then the integral in \eqref{Biot_Savart} for $v(x)$ is well-defined a.e., divergence-free in the sense of distributions as well as $\mathrm{curl}\, (v) =\w$ holds in the sense of distributions.

Notice that our result is not trivial since the procedure described at the beginning of the present section to derive the Biot-Savart law seems to require strong regularity assumptions. Indeed, solving \eqref{funkcja_pradu} with $\w$ being only $\sigma$-finite measure, in particular possibly $\w(\R^2)=\infty$, seems not trivial. Notice that due to the contribution from infinity of both $\w$ and a fundamental solution of Laplace operator, $\psi$ might not exist, and so one cannot tell that $v=-^\perp\nabla \psi$. Our result holds in a more general situation, when vorticity is a nonnegative measure which might not possess a stream function. A reader might check that actually stream function does not exist for Kaden's spirals (that are introduced in Section \ref{Kaden}). Nevertheless, we show that the velocity field can still be expressed by the Biot-Savart formula \eqref{Biot_Savart} if a measure $\w$ satisfies 
\eqref{wazne} with $\a\in (0,1)$. This way we validate the Biot-Savart law even for vorticities being $\sigma$-finite measures satisfying a condition stated already in \cite{TCMS} and being a consequence of Prandtl's similitude laws (see \cite{TCMS} and the references therein).

As to the proof of Theorem \ref{Biot-Sav}, the only non-standard part is a justification of the use of Fubini's theorem in \eqref{rrraz}. This requires a sort of potential theory type estimate. The required result is a claim of Corollary \ref{lem-setG_1}, a technical lemma yielding integrability of the integrand in \eqref{rrraz}. Consequently, see Proposition \ref{sprawdz}, the Biot-Savart law is also well-defined for vorticities for which \eqref{wazne} holds with $\a\in(0,1)$. Corollary \ref{lem-setG_1} is proven later, its proof is independent on the results of the current section. We hope it does not confuse the reader.

\begin{thm}\label{Biot-Sav}
Let $\w$ be a nonnegative $\sigma$-finite measure satisfying \eqref{mu-bdd2}. Assume that $v(x)$ is given by
\begin{equation}\label{Biot_Savartplus}
v(x)=\frac{1}{2\pi}\int_{\R^2}\frac{(x-y)^\perp}{|x-y|^2}d\w(y).
\end{equation}
Then $v\in L^1_{loc}(\R^2)$ (in particular is well-defined a.e.) and for all $\varphi \in C_0^\infty(\R^2)$ it holds
\begin{align}
\label{velocity1}
\int_{\R^2} \nabla \varphi(x)\cdot v(x) \, dx  &= 0,\\ \label{velocity2}
\int_{\R^2} {^\perp\nabla} \varphi (x)\cdot v(x) \, dx &=-\int_{\R^2}\varphi(y)\, d{\omega}(y).
\end{align}
\end{thm}

\begin{proof}
Let us prove \eqref{velocity2}. We have
\begin{align*}
&\frac{1}{2\pi}\int_{\R^2}{^\perp\nabla} \varphi (x)\cdot v(x) \, dx \\
=& \frac{1}{2\pi}\int_{\R^2}{^\perp\nabla} \varphi (x)\int_{\R^2}\frac{(x-y)^\perp}{|x-y|^2}
d\w(y) dx.
\end{align*}
First, we notice that measure $\w(x)$, satisfies assumptions of Corollary \ref{lem-setG_1}, so the latter can be applied to show that $v$ given by \eqref{Biot_Savartplus} satisfies $v\in L^1_{loc}$, in particular, $v$ is finite a.e.
On the other hand, the same Corollary \ref{lem-setG_1}, again applied to the measure $\w(x)$, allows us to use Fubini's theorem in the integral
\begin{equation}\label{rrraz}
\int_{\R^2}\int_{\R^2}\partial_i \varphi(x)\, \frac{y_j - x_j}{\abs{y - x}^2} \, d\w(y) \, dx,
\end{equation}
where $i,j \in \{1,2\}$. Indeed, there exist $R_0, M$ such that $\textrm{supp}\, \varphi \subset B(0,R_0)$ and
 $\norm{\partial_i \varphi}_\8 \le M$, it is enough to make sure that $\int_{B(0,R)}\int_\CC \frac{d\w_i(y)}{|y-x|}dx<\8$. But this is exactly \eqref{L1osz}, the main claim of Corollary \ref{lem-setG_1}.

Knowing that the Fubini theorem can be applied below, the rest of the reasoning is fully standard. We have
\begin{align*}
 \int_{\R^2} {^\perp\nabla} \varphi\cdot v \, dx &= \frac{1}{2\pi}\int_{\R^2}\int_{\R^2}{^\perp\nabla} \varphi(x)\cdot{^\perp\nabla} \ln |x-y|  d \w(y)\, dx \\
 &\stackrel{Fubini}{=} \frac{1}{2\pi}\int_{\R^2}\int_{\R^2}\nabla \varphi(x)\cdot\nabla \ln |x-y| d x\, d\w(y)=\frac{1}{2\pi}\int_{\R^2} F(y) \, d\w(y),
\end{align*}
where $F(y) = \int_{\R^2}\nabla \varphi(x)\cdot\nabla \ln |x-y| dx$. Let $\e> 0$,
    $$
    F(y) = \int_{B(y,\e)} \nabla \varphi(x)\cdot\nabla \ln |x-y| dx +\int_{\R^2 \setminus B(y,\e)} \nabla \varphi(x) \cdot\nabla \ln |x-y| dx:= F_1(y)+ F_2(y).
    $$
Obviously, $F_1(y) \leq C \e \norm{\nabla \varphi}_\8$. Denote by $\nu$ the inward normal unit vector on $\partial B(y,\e)$. By integrating by parts,
\begin{equation*}
\begin{split}
    F_2(y) &= -\int_{\R^2 \setminus B(y,\e)} \varphi(x) \Delta (\ln|x-y|)\, dx  + \int_{\partial B(y,\e)} \varphi(x) \frac{\partial}{\partial \nu} \ln|x-y|\, dl(x)\\
    &= -\e^{-1} \int_{\partial B(y,\e)} \varphi(x)\, dl(x) \to -2\pi \varphi(y).
\end{split}
\end{equation*}

The proof of \eqref{velocity1} is analogous. In the last step there we use $\frac{\partial}{\partial \nu} (\ln|x-y|)^\perp =0$. 

\hfill $\Box$
\end{proof}

Let us conclude this section with a remark concerning higher integrability of velocity $v$ given by \eqref{Biot_Savart} if $\w$ satisfies \eqref{wazne} with $\a\in(0,1)$. First, observe that in view of Proposition \ref{sprawdz}, \eqref{wazne} with $\a\in(0,1)$ implies \eqref{mu-bdd2}. Next, a consequence of Theorem \ref{mainthm1}, which we prove in Section \ref{kinen}, is a higher integrability of $v$ associated with $\w$ via \eqref{Biot_Savart}.
\begin{rem}\label{uwaga}
Let $\w$ be a nonnegative $\sigma$-finite measure which satisfies \eqref{wazne} with $\a\in(0,1)$. Then $v$ associated with $\w$ via \eqref{Biot_Savart} satisfies \eqref{velocity1} and \eqref{velocity2} and belongs to $L^2_{loc}(\R^2)$.
\end{rem}

\mysection{Kinetic energy}\label{kinen}

This section is devoted to the proof of the main result. We prove Theorem \ref{averages}, which
yields a formula representing radial averages of the square of the divergence-free velocity
associated with the vorticity being nonnegative $\sigma$-finite measure satisfying \eqref{mu-bdd2}. Moreover, we show how to infer Theorem \ref{mainthm1} from our moment representation
formula. As a consequence we obtain a precise estimate (from below and above) of the kinetic energy contained
in a ball $B(0,r)$ carried by a vorticity satisfying \eqref{wazne}.

Finally, we also state a variational problem related to the local kinetic energy estimates and find
its lower and upper bounds. Moreover, we identify the measures at which the maximal and minimal values are taken.

The next theorem provides estimates required to obtain Theorem \ref{mainthm1} as a consequence of Theorem \ref{averages}. It also gives very precise constants in the bounds of local kinetic energy which will be used further to identify minimizers as well as maximizers of a variational problem leading to the local kinetic energy estimates.

\begin{thm}\label{averages2}
Let $\w$ be a nonnegative $\sigma$-finite measure on $\CC$ that satisfies \eqref{wazne} with $\a\in(0,1)$ and $c>0$. Then we obtain
\begin{equation}
\label{eq-av2}
   \frac{c^2}{4\pi^2}r^{2\a-2} \leq  A_r(\w) \leq   \frac{ c^2\a^2}{4\sin^2(\pi \a)}  r^{2\a-2}.
\end{equation}
\end{thm}

\proof{
We shall use Theorem \ref{averages}. For the lower estimate, observe that
$$
    4\pi^2 A_r(\w) \geq r^{-2} m_{r,0}(\w)^2 = r^{-2} \w(B(0,r))^2 = c^2r^{2\a-2}.
$$

For the upper estimate, we notice that \eqref{wazne} allows us to estimate the moments,
\begin{align}
    \label{est_m}
    \abs{m_{r,n}(\w)} &\leq \int_{B(0,r)} |u|^n \, d\w(u) = c\int_0^r s^{n} \a s^{\a-1}\, ds = c\frac{\a}{n+\a} r^{n+\a} \quad (n\geq 0),\\
    \label{est_M}
    \abs{M_{r,k}(\w)} &\leq \int_{\CC\setminus B(0,r)} |u|^{-k}\, d\w(u) = c\int_r^\8 s^{-k} \a s^{\a-1}\, ds = c\frac{\a}{k-\a} r^{-k+\a} \quad (k\geq 1).
\end{align}
This leads to
\begin{equation*}
\begin{split}
    4\pi^2 A_r(\w) &\leq c^2\sum_{n=0}^\8 r^{-2n-2}\left(\frac{\a}{n+\a}\right)^2 r^{2n+2\a} + c^2\sum_{k=1}^\8 r^{2k-2}\left(\frac{\a}{k-\a}\right)^2 r^{-2k+2\a}\\
    &= c^2\a^2 \sum_{n=-\8}^\8 \frac{1}{(n+\a)^2}r^{2\a-2} = c^2\frac{\pi^2 \a^2}{\sin^2(\pi \a)}  r^{2\a-2}.
\end{split}
\end{equation*}
\hfill $\Box$
}

We finish this section with the proof of Theorem \ref{mainthm1} provided that Theorem \ref{averages} holds. The latter is proven in Section \ref{srednie}.

\noindent
\textbf{Proof of Theorem \ref{mainthm1}.}
We recall \eqref{radial} and see that the kinetic energy $E_r(\w)$ is given as $2\pi\int_0^r sA_s(\w)ds$. Hence, integrating the bounds in \eqref{eq-av2} in $r$, we obtain the claim of Theorem \ref{mainthm1}.

\vspace{0.3cm}
\hfill $\Box$

\subsection{Variational formulation and its minimizer and maximizer} \label{sec_ex}
In what follows we introduce a functional over a certain subset of $\sigma$-finite measures which gives the value of
local kinetic energy associated with these measures. We identify the possible extreme values and provide the extremizers.  Let us fix $\a\in(0,1)$. Then we define a set ${\cal A}$ as
\[
{\cal A}:=\{\mbox{nonnegative $\sigma$-finite measures $\w$ satisfying}\; \eqref{wazne}\;\mbox{with}\; \a\in (0,1) \}.
\]
Fix $r>0$ and using \eqref{def_energy} define the functional $E_r(\w)$ for any $\w\in{\cal A}$. We look for its minimum and maximum. The lower and upper bounds are given in Theorem \ref{mainthm1}. We show that those are actually
achieved and provide the examples of extremals.

Before proceeding with the argument let us notice that in the case of slightly more regular velocities
such a functional was used in literature to construct steady states of the incompressible 2d Euler system, see
\cite{arnold} in the case of regular solutions and \cite{turkington} in the case of vortex patches. It is related to the hamiltonian structure of the Euler system. It is not clear to us whether the same approach could work in the case of steady vortex sheets (like those introduced in \cite{DiPerna_majda}).

Denote by $\w_\infty$ and $\w_0$ respectively
\begin{equation}\label{w_infty}
    d\w_\infty(x_1+ix_2) = \frac{c\a}{2\pi} |x_1+ix_2|^{\a-2} \, dx_1\, dx_2,
\end{equation}
\begin{equation}\label{w_0}
    d\w_0(x_1+ix_2) = c\a x_1^{\a-1} \chi_{(0,\8)}(x_1)\delta_{0}(x_2)\, dx_1\, dx_2.
\end{equation}
The first one is a radially symmetric measure and the latter one is a vortex sheet supported on the half-line. In both cases they are chosen so that \eqref{wazne} holds.

First we notice that all the moments for $\w_\infty$ (except $m_{r,0}(\w_\infty)$) vanish.  Thus
$$
    A_r(\w_\infty) =  \frac{c^2}{4\pi^2}r^{2\a-2} \qquad \text{and} \qquad E_r(\w_\infty) = \frac{c^2}{{4\pi\a}} r^{2\a}.
$$
From the proof of Theorem \ref{averages2} one observes that for $\w_0$ all the upper estimates become equalities, thus
$$
    A_r(\w_0) = \frac{c^2\a^2}{4\sin^2(\pi\a)}  r^{2\a-2} \qquad \text{and} \qquad E_r(\w_0) = \frac{c^2\a \pi }{4\sin^2(\pi \a)}\, r^{2\a}.
$$
This proves that the estimates given in \eqref{mainthm1est} and \eqref{eq-av2} are optimal, meaning that $w_\infty$ and $\w_0$ are minimizer and maximizer of $E_r$ over ${\cal A}$, respectively.

\subsection{Proof of Theorem \ref{averages}}\label{srednie}

This subsection is essentially self-contained and can be read independently of the rest of the paper. We shall use the complex notation, that is $z = x+iy \in \CC = \R^2$. The goal of the section is to prove Theorem \ref{averages}. Assume that a nonnegative measure $\w$ on $\CC$ is given such that \eqref{mu-bdd2} is satisfied.

In the complex notation the velocity is given by
\begin{equation}
\label{eq-v}
    v(z) = \frac{1}{2\pi} \int_\CC \frac{i(z-u)}{|z-u|^2} \, d\w(u) = \frac{1}{2\pi} \int_\CC \frac{i}{\ov{z-u}} \, d\w(u).
\end{equation}
We assume only \eqref{mu-bdd2}. It turns out that this is enough to guarantee that \eqref{eq-v} is well-defined for a.e. $z\in \CC$. As we shall see in Corollary \ref{lem-setG_1}, $v\in L^1_{loc}$ and so $v$ is finite a.e. Moreover, we have seen in Section \ref{vv} that Corollary \ref{lem-setG_1} is an important factor of the proof of the Biot-Savart formula for vortex sheets satisfying \eqref{mu-bdd2}.

The proof of Theorem \ref{averages} splits into several steps. Let us start with the following definition, set
$$
    G = \set{r>0 \ : \ \int_{\CC} \frac{d\w(u)}{\sqrt{\abs{r^2 - |u|^2}}} < \8}.
$$
In particular, $\w(rS^1)=0$ for $r\in G$, where $S^1 = \{z\in \CC \ : \ |z|=1\}$. We show below that this set is of full measure in $(0,\infty)$. It will be essential in showing that the divergence-free velocity associated with $\w$ is well-defined a.e..
\begin{lem}\label{lem-setG}
Let $\w$ be a nonnegative $\sigma$-finite measure satisfying \eqref{mu-bdd2}. Then the Lebesgue measure of $G^{c} = (0,\8) \setminus G$ is zero.
\end{lem}

\proof{
For $R\geq 1$ we shall prove that
\begin{equation}\label{convergence}
    W=\int_0^R r \int_\CC \frac{d\w(u)}{\sqrt{\abs{r^2-|u|^2}}} \, dr <\8,
\end{equation}
which obviously implies the lemma. Let
$$
    W= W_1+W_2 = \int_{|u|<R} \int_0^R  \frac{r\, dr}{\sqrt{\abs{r^2-|u|^2}}} \, d\w(u) + \int_{|u|\geq R} \int_0^R  \frac{r\, dr}{\sqrt{\abs{r^2-|u|^2}}} \, d\w(u).
$$
By Fubini's theorem and \eqref{mu-bdd2} we have
\begin{equation*}
\begin{split}
    W_1 &= \int_{|u|<R} \int_0^{|u|}  \frac{r\, dr}{\sqrt{|u|^2-r^2}} \, d\w(u) + \int_{|u|<R} \int_{|u|}^R  \frac{r\, dr}{\sqrt{r^2-|u|^2}} \, d\w(u)\\
        &=  \int_{|u|<R} \left(|u|+\sqrt{R^2-|u|^2}\right) \, d\w(u)\\
        &\leq 2R\w(B(0,R)) <\8.
\end{split}
\end{equation*}
Similarly,
\begin{equation*}
\begin{split}
    W_2 &= \int_{|u|\geq R} \left(|u|-\sqrt{|u|^2-R^2}\right)\, d\w(u) = \int_{|u|\geq R} \frac{R^2}{|u|+\sqrt{|u|^2 - R^2}}\, d\w(u)\\
     &\leq R^2 \int_{|u|\geq R} \frac{d\w(u)}{|u|}\leq 2 R^2 \int_{\CC} \frac{d\w(u)}{1+|u|} <\8,
\end{split}
\end{equation*}
where in the last inequality we have used that $R\geq 1$.

\hfill $\Box$
}

The next lemma is crucial in our investigation. It uncovers the essential cancellations separating the outer and inner contribution to the spherical averages of kinetic energy.
\begin{lem}\label{lem-residua}
Let $u,v \in \CC$, $r>0$, $|u|,|v| \neq r$. Then
\begin{equation}
    (2\pi)^{-1} \int_0^{2\pi} \frac{d\theta}{(u-re^{i\theta})(v-re^{-i\theta})} =
    \begin{cases}
    (uv-r^2)^{-1} \quad & |u|, |v| >r,\\
    (r^2-uv)^{-1} \quad & |u|, |v| <r,\\
    0 \quad & (|u|-r)(|v|-r)<0 .
    \end{cases}
\end{equation}
\end{lem}

\proof{
In the region  $|z|<r$ we have
\begin{equation}\label{exp1}
    \frac{1}{z-re^{i\theta}} = -\frac{1}{re^{i\theta}}\frac{1}{1-\frac{z}{re^{i\theta}}}=-\sum_{n=1}^\8 \frac{z^{n-1}}{r^n} e^{-in\theta},
\end{equation}
whereas for $|z|>r$ we have
\begin{equation}\label{exp2}
    \frac{1}{z-re^{i\theta}} = \sum_{n=0}^\8 \frac{r^n}{z^{n+1}} e^{in\theta}.
\end{equation}
Consider first the case $|u|,|v|<r$. Then in light of \eqref{exp1} one obtains
\begin{equation*}
\begin{split}
    (2\pi)^{-1} \int_0^{2\pi} \frac{d\theta}{(u-re^{i\theta})(v-re^{-i\theta})} &= (2\pi)^{-1}\sum_{n,m=1}^\8 \frac{u^{n-1} v^{m-1}}{r^{n+m}}  \int_0^{2\pi} e^{i(m-n)\theta} \,d\theta \\
    &= \sum_{n=1}^\8 \frac{(uv)^{n-1}}{r^{2n}} = \frac{1}{r^2-uv}\;.
\end{split}
\end{equation*}
In the case when $|u|, |v|>r$ the proof follows similarly using \eqref{exp2}.

Finally, let us consider the case $|u|<r, |v|>r$. Here we observe the crucial cancellations.
\[
(2\pi)^{-1} \int_0^{2\pi} \frac{d\theta}{(u-re^{i\theta})(v-re^{-i\theta})}=-\sum_{n=1,m=0}^{\8}\frac{u^{n-1}}{r^n}\frac{r^m}{v^{m+1}}\int_0^{2\pi}e^{-i(n+m)\theta}\,d\theta.
\]
Since for $n\geq 1, m\geq 0$ we have $n+m>0$, the integral on the right-hand side of the above identity is zero. Hence the claim follows.

\hfill $\Box$
}

Let us notice that Lemma \ref{lem-residua} can also be proved by the residue theorem. In a special case $v=\ov{u}$  we get the following.

\begin{cor}\label{cor1}
For $|u|\neq r$ we have
$$
    (2\pi)^{-1} \int_0^{2\pi} \frac{d\theta}{\abs{u-re^{i\theta}}^2} = \abs{r^2-|u|^2}^{-1}.
$$
\end{cor}

Moreover, using the Cauchy-Schwarz inequality, Corollary \ref{cor1}, and the definition of the set $G$, we arrive at the following.

\begin{cor}\label{cor2}
For a nonnegative $\sigma$-finite measure $\w$, which satisfies \eqref{mu-bdd2}, and $r\in G$ we have
$$
     \int_\CC \int_\CC  \int_0^{2\pi} \frac{d\theta}{\abs{u-re^{i\theta}}\abs{\ov{w}-re^{-i\theta}}} \, d\w(u) \, d\w(w) <\8.
$$
\end{cor}

Corollary \ref{cor2} is essential in the proof of Theorem \ref{averages}. It justifies an application of the Fubini theorem in a crucial moment.

Before proceeding with the proof of Theorem \ref{averages}, let us state the next corollary, which on the one hand, guarantees that $v$, as defined in \eqref{eq-v}, is finite a.e., on the other hand gives a strong estimate
which is used in Section \ref{vv} to state a general version of the Biot-Savart law.

\begin{cor}\label{lem-setG_1}
Let $R>0$. For a nonnegative $\sigma$-finite measure $\w$, which satisfies \eqref{mu-bdd2}, there holds
\begin{equation}\label{L1osz}
     \int_{B(0,R)} \int_\CC \frac{d\w(u)}{|u-z|} \, dz <\8.
\end{equation}
In particular, $v$ defined by \eqref{eq-v}, satisfies $v\in L^1_{loc}$ and so $v<\8$ a.e..
\end{cor}
\proof{
By H\"{o}lder's inequality we have
\[
\int_0^{2\pi}\frac{d\theta}{|u-re^{i\theta}|}\leq (2\pi)^{1/2}\left(\int_0^{2\pi}\frac{d\theta}{|u-re^{i\theta}|^2}\right)^{1/2},
\]
hence by Fubini's theorem
\[
\int_{B(0,R)} \int_\CC \frac{d\w(u)}{|u-z|} \, dz\leq\int_0^R r\int_\CC (2\pi)^{1/2}\left(\int_0^{2\pi}\frac{d\theta}{|u-re^{i\theta}|^2}\right)^{1/2}d\w(u)dr.
\]
According to Corollary \ref{cor1}, the latter equals
\[
\int_0^R r\int_\CC 2\pi \frac{d\w(u)}{\sqrt{|r^2-|u|^2|}}dr,
\]
which is finite by \eqref{convergence}. Hence, \eqref{L1osz} is proven,
which immediately guarantees that $v\in L^1_{loc}$.

\hfill $\Box$
}

\vspace{0.5cm}
\noindent
\textbf{Proof of Theorem \ref{averages}.}
Assume that $r\in G$. By \eqref{eq-v},
\begin{equation*}
\begin{split}
    A_r(\w)& = (2\pi)^{-1} \int_0^{2\pi} \abs{v(re^{i\theta})}^2 \, d\theta =  (2\pi)^{-3} \int_0^{2\pi} \abs{\int_\CC \frac{i d\w(u)}{re^{-i\theta}-\ov{u}}}^2 \, d\theta\\
    &= (2\pi)^{-3} \int_\CC \int_\CC \int_0^{2\pi} \frac{1}{(u-re^{i\theta})(\ov{w}-re^{-i\theta})}\, d\theta\, d\w(u) \, d\w(w).
\end{split}
\end{equation*}
In the last equality we have used Fubini's theorem, see Corollary \ref{cor2}. Applying Lemma \ref{lem-residua} we obtain
\begin{equation*}
\begin{split}
    4\pi^2 A_r(\w) =& \int_{B(0,r)} \int_{B(0,r)} \frac{1}{r^2-u\ov{w}}\, d\w(u)\, d\w(w) + \int_{\CC \setminus B(0,r)} \int_{\CC \setminus B(0,r)} \frac{1}{u\ov{w} - r^2} \, d\w(u)\, d\w(w)\\
    =& \int_{B(0,r)} \int_{B(0,r)} \sum_{n=0}^\8 \frac{u^n \ov{w}^n}{r^{2n+2}}\, d\w(u)\, d\w(w) \\
    &+ \int_{\CC \setminus B(0,r)} \int_{\CC \setminus B(0,r)} \sum_{k=0}^\8 \frac{r^{2k}}{u^{k+1} \ov{w}^{k+1}} \, d\w(u)\, d\w(w)\\
    =&\sum_{n=0}^\8 r^{-2n-2} m_{r,n}(\w) \ov{m_{r,n}(\w)} + \sum_{k=1}^\8 r^{2k-2} M_{r,k}(\w) \ov{M_{r,k}(\w)},
\end{split}
\end{equation*}
This ends the proof of Theorem  \ref{averages}, provided that we justify the last equality. It suffices to have
\[
I_{1}=\int_{B(0,r)}\int_{B(0,r)}\sum_{n=0}^\8 \frac{|u|^n |w|^n}{r^{2n+2}}\, d\w(u)\, d\w(w) < \8
\]
and
\[
I_{2}=\int_{\CC \setminus B(0,r)} \int_{\CC \setminus B(0,r)} \sum_{k=0}^\8 \frac{r^{2k}}{|u|^{k+1} |w|^{k+1}} \, d\w(u)\, d\w(w) < \8.
\]
To this end, notice that
\begin{equation*}
\begin{split}
I_{1}&=\int_{B(0,r)}\int_{B(0,r)} \frac{d\w(u)\, d\w(w)}{r^{2}-|u||w|}
\leq \int_{B(0,r)}\int_{B(0,r)} \frac{d\w(u)\, d\w(w)}{\sqrt{r^{2}-|u|^{2}}\sqrt{r^{2}-|w|^{2}}}\\&
=\left( \int_{B(0,r)}\frac{d\w(u)}{\sqrt{r^{2}-|u|^{2}}}\right)^{2},
\end{split}
\end{equation*}
which is finite since $r \in G$. We have used the fact that $(r^{2}-|u||w|)^{2} \geq (r^{2}-|u|^{2})(r^{2}-|w|^{2})$.
In a~similar way,
\begin{equation*}
\begin{split}
I_{2}&=\int_{\CC \setminus B(0,r)} \int_{\CC \setminus B(0,r)} \frac{d\w(u)\, d\w(w)}{|u||w|-r^{2}} \leq
\int_{\CC \setminus B(0,r)} \int_{\CC \setminus B(0,r)} \frac{d\w(u)\, d\w(w)}{\sqrt{|u|^{2}-r^{2}}\sqrt{|w|^{2}-r^{2}}}\\&
=\left( \int_{\CC \setminus B(0,r)} \frac{d\w(u)}{\sqrt{|u|^{2}-r^{2}}}\right)^{2} < \infty.
\end{split}
\end{equation*}

\hfill $\Box$
\subsection{Extension to vorticities}

For future issues let us at this point extend the moment formula to {\it vorticities}. It will be useful when studying limits of Kaden's spiral when time approaches $0$ or $\infty$. It seems to us that it can be applicable in many other situations concerning convergence or simply computation of the difference of velocities based on the difference of their vorticities.

Consider $\tilde{\w}=\w_1-\w_2$, where $\w_i$ satisfy \eqref{mu-bdd2}. Notice that all the moments \eqref{moment0}--\eqref{moment2} for either $\w_1$ or $\w_2$ are finite, thus we can define same moments for $\tilde{\w}$ as a difference of the moments for $\w_1$ and $\w_2$. Completely the same argument as in Lemma
\ref{lem-setG} leads us to the claim that the set $G_{sign}=G_1\cap G_2\subset (0,\infty)$,
\[
G_i= \set{r>0 \ : \ \int_{\CC} \frac{d\w_i(u)}{\sqrt{\abs{r^2 - |u|^2}}} < \8}\;,\; i=1,2,
\]
is of full measure.

Next, proceeding in the same way as in the proof of Corollary \ref{cor2}, we arrive at
\begin{cor}\label{cor3}
Let $\tilde{\w}=\w_1-\w_2$ be such that \eqref{mu-bdd2} is satisfied. Assume moreover that $r\in G_{sign}$. Then for $k,l\in\{1,2\}$ we have
$$
     \int_\CC \int_\CC  \int_0^{2\pi} \frac{d\theta}{\abs{u-re^{i\theta}}\abs{\ov{w}-re^{-i\theta}}} \, d\w_k(u) \, d\w_l(w) <\8.
$$
\end{cor}
Corollary \ref{cor3} allows us to use Fubini's theorem in the present context and this way extend Theorem
\ref{averages} to its version for {\it vorticities}. First, notice that due to linearity in $\w$ of the Biot-Savart operator as well as $m_{r,n}$ and $M_{r,k}$, one immediately sees how to understand quantities occurring in \eqref{averages_main}. Moreover, by linearity, the proof of Theorem \ref{averages} goes in the same way also for $\tilde{\w}$. The following holds.
\begin{cor}\label{3.8}
Let $\tilde{\w}=\w_1-\w_2$, where nonnegative $\sigma$-finite measures $\w_1$ and $\w_2$ satisfy \eqref{mu-bdd2}. Then $\w$ satisfies \eqref{averages_main}.
\end{cor}

\mysection{Applications - Kaden spirals}\label{Kaden}

In the present section our aim is to introduce and study some properties of Kaden's spirals. It turns out that the framework of moment formula we introduced in Theorem \ref{averages} is very helpful in this respect. In order to define the Kaden spiral we need to first consider the Birkhoff-Rott equation. It was introduced in \cite{birkhoff, rott} independently and it describes the time evolution of the curve $\cv$ -an interface between the zero vorticity regions
of the vortex sheet whose velocity field attains the tangential velocity discontinuity along the curve $\cv$.
Let us denote the position of the curve $\cv$ at time $t$ and cumulative vorticity $\Gamma$ by $Z(\Gamma, t)$. Then the following equation is satisfied (we refer an interested reader to the handbook \cite{mar_pul} for details).
$$
\frac{\partial}{\partial t}\overline{Z(\Gamma,t)} = \frac{1}{2\pi i} \textrm{p.v.}\int_{\R} \frac{d\Gamma'}{Z(\Gamma,t)-Z(\Gamma',t)},
$$
where $Z\, : \, \R \times (0,\8) \to \CC$. For $\mu >0$ we look for self similar solutions of the form
\begin{equation}\label{selfsimilar}
Z(\Gamma,t) = t^\mu z(\gamma), \quad \quad \gamma = t^{1-2\mu}\Gamma.
\end{equation}
Such solution in the new variables $\gamma \in \R$ and $t>0$ satisfies the equation
\begin{equation}\label{BR2}
(1-2 \mu)\gamma \,  z'(\gamma) + \mu z(\gamma) = \frac{i}{2\pi} \textrm{p.v.} \int_{\R} \frac{d\gamma'}{\overline{z(\gamma)-z(\gamma')}}\;.
\end{equation}

The construction of the Kaden spirals, introduced in \cite{kaden} and reviewed recently in \cite{elling}, is based on the following ansatz.
Arcs of spirals appearing in nature, when packed densely, are similar to arcs of a circle. Thus each arc (each $2 \pi$ turn around the origin) can be approximated by a circle. Now assume that the whole vorticity is concentrated on a circle, then if $P$ is the point inside the circle, then the velocity in $P$ generated by the vorticity is zero, while for any point lying outside the circle the velocity is the same as the velocity generated by the point vortex placed at the centre of the circle with the mass/strength/vorticity equal to the total vorticity of the circle.
Using the above approach we assume that the part of the spiral further away to the origin than considered point $z(\gamma)$ does not influence the velocity in the point, while the part closer to the origin than $z(\gamma) $ extorts the velocity
equal to
$
\frac{i\gamma}{2 \pi \cdot \overline{z(\gamma)}}.
$

Inserting conclusion of the above heuristics into \eqref{BR2} we get
$$
    (1-2 \mu)\gamma z'(\gamma) + \mu z(\gamma) = \frac{i\gamma}{2\pi \overline{z(\gamma)}}.
$$
Switching to polar coordinates, $z(\gamma) = r(\gamma)e^{i\theta(\gamma)}$, we get
$$
    (1-2\mu)\gamma\big(r'(\gamma)e^{i\theta(\gamma)} + i\theta'(\gamma) r(\gamma) e^{i\theta(\gamma)} \big) +
    \mu r(\gamma) e^{i\theta(\gamma)}= \frac{i \gamma}{2 \pi r(\gamma)e^{- i\theta(\gamma)}}.
$$
By dividing by $e^{i\theta(\gamma)}$ and separating the real and imaginary part we obtain a system of ordinary differential equations for $r$ and $\theta$, namely
\begin{equation} \label{BR5}
    \left\{ \begin{array}{l}
    {(1-2\mu)\gamma r'(\gamma) + \mu r(\gamma) = 0}, \\
    {(1-2\mu)\gamma r(\gamma) \theta'(\gamma) = \frac{\gamma}{2\pi r(\gamma)}}.
    \end{array} \right.
\end{equation}
For $\mu \neq 1/2$ the solution of \eqref{BR5} is
$$
    \left\{ \begin{array}{l}
    r(\gamma)= C_1 |\gamma|^\frac{\mu}{2 \mu - 1},\\
    \theta(\gamma) = \frac{1}{2\pi C_1^2} |\gamma|^\frac{1}{1-2\mu}+C_2.
    \end{array} \right.
$$

Further on we take $C_1=1, C_2=0$ and consider only $\gamma>0$ restricting our consideration to one representative instead of a family of curves. In polar coordinates this is a spiral given by
$$
    \theta(r) = \frac{1}{2\pi} r^{-1/\mu}, \qquad r>0.
$$

We are interested in an evolution in time of the vorticity and the velocity field. Thus we go back to the original variable $\Gamma$ (see \eqref{selfsimilar}) getting
\[
Z(\Gamma,t) = t^\mu z(t^{1-2\mu}\Gamma) = R(\Gamma,t) e^{i\Theta(\Gamma,t)} ,
\]
with
\begin{equation} \label{RofGamma}
\left\{ \begin{array}{l}
R(\Gamma,t)= \Gamma^{\frac{\mu}{2\mu-1}},\\
\Theta(\Gamma,t) = \frac{t}{2\pi} \Gamma^{\frac{1}{1-2\mu}}.
\end{array} \right.
\end{equation}

We will denote the spiral curve that is the support of the vorticity for a given time $t>0$, by $\cv_t$. In polar coordinates $\cv_t$ is given by the equation
\begin{equation}\label{polar}
\Theta(R) = \frac{t}{2\pi}  R ^{-\frac{1}{\mu}}.
\end{equation}

The measure corresponding to the vorticity for the Kaden spiral $\cv_t$ at time $t>0$ will be denoted by $\omega_t$. The support of $\w_t$ is $\cv_t$, moreover since $\Gamma$ is a cumulative vorticity in a ball of radius $|Z|$
\[
\omega_t\left(B(0,|Z(\Gamma,t)|)\right) = \Gamma,
\]
or, equivalently,
\begin{equation}\label{omega-ball-2}
    \w_t\left(B(0,r)\right) = r^{2-\frac{1}{\mu}}.
\end{equation}

Notice that the spiral concentrates at the origin and for $\mu\in (1/2,1)$ the length of $\cv_t\cap B(0,r)$ is infinite. Indeed, the following fact holds.

\begin{prop}\label{lenght}
For each $t>0$ and $\mu\in(1/2,1)$ the Kaden spiral $\cv_t$ restricted to a ball $B(0,r)$, $r>0$, has infinite length.
\end{prop}

\proof{
Let us parametrize the Kaden spiral \eqref{polar}, by
    $
    (0,\8) \ni s \mapsto s \exp\left(\frac{it}{2\pi} s^{-1/\mu}\right).
    $
Then the length of $\cv_t$ on $B(0,r)$ is given by
$$
    l\left(\cv_t\right) = \int_0^r \sqrt{1 + s^2 \frac{t^2}{(2\pi \mu)^2} s^{-2-2/\mu}}\, ds \geq \frac{t}{2\pi \mu} \int_0^r s^{-1/\mu}\, ds=\8
$$
for $\mu\in(1/2,1)$.

\hfill $\Box$
}

\begin{lem}\label{density}
Let $\cv_t$ be the Kaden spiral with $\mu\in(1/2,1)$. Then for $f\in L^1(\R^2, \w_t)$ we have
$$
    \int_{\R^2} f(x) \, d\w_t(x) = \left(2-\frac{1}{\mu}\right)\int_0^\8 f\left(s\exp\left(\frac{it}{2\pi} s^{-1/\mu}\right)\right)s^{1-1/\mu}\, ds.
$$
\end{lem}

\proof{
We shall use the same parametrization of $\cv_t$ as in the proof of Proposition \ref{lenght}. As we know that $\w_t$ is supported on $\cv_t$ and \eqref{omega-ball-2} holds, we can find a non-negative density $g_t(s)$, such that (see \cite{rudin})
$$
    \int_{\R^2} f(x) \, d\w_t(x) = \int_0^\8 f\left(s\exp\left(\frac{it}{2\pi} s^{-1/\mu}\right)\right)g_t(s)\, ds.
$$
By taking $f=\chi_{B(0,r)}$ and using \eqref{omega-ball-2} again, we get
$$
    r^{2-1/\mu} = \int_0^r g_t(s)\, ds.
$$
Differentiating both sides we obtain $g_t(s) = (2-1/\mu) s^{1-1/\mu}$.

\hfill $\Box$
}

We end this section with a simple observation concerning kinetic energy of the considered velocity field.
\begin{prop}\label{rem_scaling}
Consider the velocity field generated by the Kaden spiral $\cv_t$ via \eqref{Biot_Savart}. Then we have the following scaling of the kinetic energy in a ball $B(0,r)$,
\[
E_r(\w_t) = t^{4\mu-2}E_{t^{-\mu}r}(\w_1).
\]
\end{prop}

\proof{
Using Lemma \ref{density} and changing variables twice ($R \mapsto R t^\mu$ and $x \mapsto x t^\mu$),
\begin{equation*}
\begin{split}
4\pi^2E_r(\w_t)&= \left(2-\frac{1}{\mu}\right)^2 \int_{B(0,r)}\abs{\int_0^\infty \frac{\left(x-R\exp\left(\frac{it}{2\pi} R^{-1/\mu}\right)\right)^\perp}{\abs{x-R\exp\left(\frac{it}{2\pi} R^{-1/\mu}\right)}^2} R^{1-\frac{1}{\mu}} dR }^2\,dx\\
&= \left(2-\frac{1}{\mu}\right)^2 t^{4\mu-2} \int_{B(0,r)}\abs{\int_0^\infty \frac{\left(x-Rt^\mu \exp\left(\frac{i}{2\pi} R^{-1/\mu}\right)\right)^\perp}{\abs{x-Rt^{\mu}\exp\left(\frac{i}{2\pi} R^{-1/\mu}\right)}^2} R^{1-\frac{1}{\mu}} dR }^2\,dx\\
&= \left(2-\frac{1}{\mu}\right)^2 t^{4\mu-2} \int_{B(0,t^{-\mu}r)}\abs{\int_0^\infty \frac{\left(x-R \exp\left(\frac{i}{2\pi} R^{-1/\mu}\right)\right)^\perp}{\abs{x-R\exp\left(\frac{i}{2\pi} R^{-1/\mu}\right)}^2} R^{1-\frac{1}{\mu}} dR }^2\,dx\\
&=4\pi^2 t^{4\mu-2} E_{t^{-\mu}r}(\w_1).
\end{split}
\end{equation*}

\hfill $\Box$
}

\subsection{End point estimates of the energy}

The present subsection is devoted to study time evolution of the velocity field carried by Kaden's spiral, in particular we can estimate the evolution of the kinetic energy $E_r(\w_t)$ of Kaden's spiral, when $r>0$ is fixed and $t$ goes either to zero or to infinity. The constants $\a$ and $\mu$ are always related by
\begin{equation}\label{amu}
\a= 2-1/\mu,
\end{equation}
see \eqref{wazne}, \eqref{omega-ball-2} for the definitions of $\a$ and $\mu$. We assume that $\mu\in (1/2,1)$. Recall $\w_0$ and $\w_\8$, the measures defined in Subsection \ref{sec_ex}.

The results presented below show the applicability of our main Theorem \ref{averages} in the examination of Kaden's spiral. On the one hand we show that the velocity field related to Kaden's spiral, by the results of Theorem \ref{Biot-Sav} of Section \ref{vv}
such exists and is an element of $L^2_{loc}$, dissipates the energy in any ball surrounding the origin of a spiral. Indeed, we prove below that when time approaches zero, kinetic energy contained in a ball centered in an origin of Kaden's spiral tends to the maximal possible value of local kinetic energy carried by the vorticity satisfying \eqref{wazne} with $\a\in (0,1)$. When time tends to $\infty$, then the kinetic energy in a ball centered in an origin approaches the minimal possible value. This means that in the meantime the energy is pushed out from any ball surrounding the origin of the Kaden approximation, the latter indicates a sort of viscosity in the center of the spiral.

On the other hand, we show that velocity field associated with Kaden's spiral converges in $L^2_{loc}$ to
\eqref{w_0} with $c=1$ when time tends to $0$. This convergence is interesting in view of the problem of uniqueness of Delort's solutions of 2d Euler equation constructed in \cite{delort}. Such solutions have vorticities being compactly supported nonnegative Radon measures. Whether they are unique is still an open question. If the requirement that vorticities are measures is relaxed, it is known that there are infinitely many vortex sheet solutions satisfying the 2d Euler equations, see \cite{szekelyhidi}. However, velocity fields constructed in \cite{szekelyhidi} are extremely oscillating, so that their vorticities are not even measures. Numerical simulations suggest that spirals of vorticity could be the counterexamples to the uniqueness problem in the Delort's class of solutions with vorticities being measures. It is observed in the computations that such spirals detach from the steady solution
of the form similar to that given by \eqref{w_0}, see for instance \cite{lop_low_nus_zheng, pullin}. Hence our result concerning the convergence of Kaden's spiral with time approaching $0$ is of interest in this respect, in particular since Kaden's spiral is a nonnegative measure.

Finally, let us notice that the evolution of the Kaden spiral is a path in the class of nonnegative $\sigma$-finite measures linking the object with maximal value of the energy functional with the one with minimal value (see Subsection \ref{sec_ex}).

\begin{prop}\label{prop_infty}
Assume that $r>0$ is fixed. Let $\w_t$ be the vorticity of the Kaden spiral defined in \eqref{RofGamma} with $\mu\in
(1/2,1)$
and $\w_\infty$ be given by \eqref{w_infty} with $c=1$, $\a$ related to $\mu$ via \eqref{amu}, in particular $\a\in(0,1)$. Next, consider $u(t)$ and $u_\infty$ as divergence-free velocity fields related to $\w_t$ and $\w_\infty$, respectively, via the Biot-Savart law.
Then
    \begin{equation}\label{est_prop_infty}
    \lim_{t\to \8} \int_{B(0,r)}|u(t)-u_\infty|^2dx=0.
    \end{equation}
In particular,	
    \begin{equation}\label{est_prop_zero5}
    \lim_{t\to \infty} E_r(\w_t) = E_r(\w_\infty).
    \end{equation}
\end{prop}

\begin{proof}

First, we notice that $m_{r,0}(\w_t) = m_{r,0}(\w_\8)$. Indeed, $\w_\8$ satisfies \eqref{wazne} with $c=1$, $\a=2-1/\mu$ and $\w_t$ satisfies \eqref{omega-ball-2}. Next, as we have already observed, all the other moments for $\w_\8$ vanish. Hence, using Corollary \ref{3.8}, we arrive at
		\begin{eqnarray}\label{warunek}
4\pi^2|A_r(\w_t-\w_\8)|&=&\sum_{n=0}^{\8}r^{-2n-2}|m_{r,n}(\w_t-\w_\8)|^2+\sum_{k=1}^{\8}r^{2k-2}|M_{r,k}(\w_t-\w_\8)|^2\nonumber\\
&=&	\sum_{n=0}^{\8}r^{-2n-2}|m_{r,n}(\w_t)-m_{r,n}(\w_\8)|^2+\sum_{k=1}^{\8}r^{2k-2}|M_{r,k}(\w_t)-M_{r,k}(\w_\8)|^2\nonumber\\
&=&\sum_{n=1}^{\8}r^{-2n-2}|m_{r,n}(\w_t)|^2+\sum_{k=1}^{\8}r^{2k-2}|M_{r,k}(\w_t)|^2.
		\end{eqnarray}

Using Lemma \ref{density} and integration by parts, for $n\geq 1$ we get
    \begin{equation*}
    \begin{split}\label{est47}
    m_{r,n}(\w_t) &= \a \int_0^r s^{n+1-1/\mu} \exp\left(\frac{itn}{2\pi} s^{-1/\mu}\right)\, ds\cr
    &=-\frac{2\pi\mu \a}{itn} \int_0^r  s^{n+2} \, \left(\exp\left(\frac{itn}{2\pi} s^{-1/\mu}\right)\right)'\, ds\cr
    &=-\frac{2\pi \mu \a}{itn} \left[r^{n+2} \exp\left(\frac{itn}{2\pi} r^{-1/\mu}\right) -
    (n+2)  \int_0^r s^{n+1} \, \exp\left(\frac{itn}{2\pi} s^{-1/\mu}\right)\, ds\right]
    \end{split}
    \end{equation*}
and, bounding the absolute value of both exponential terms by 1,
    \begin{equation}\label{est48}
    \abs{m_{r,n}(\w_t)} \leq C \frac{r^{n+2}}{tn}.
    \end{equation}

Now we proceed to estimate the outer moments. We shall consider $M_{r,1}$ and $M_{r,2}$ separately. By changing variables $ t s^{-1/\mu}=a$ we get
    \begin{equation*}\label{est49}
    \begin{split}
    M_{r,1}(\w_t) &= \a \int_r^\8 s^{-1/\mu} \exp\left(-\frac{it}{2\pi} s^{-1/\mu}\right)\, ds\cr
    &= \a\mu t^{\mu-1} \int_0^{t\, r^{-1/\mu}} a^{-\mu} \exp\left(-\frac{ia}{2\pi}\right)\, da.
    \end{split}
    \end{equation*}
Since the integral $\int_0^\8 a^{-\mu} \exp(-\frac{ia}{2\pi})\, da$ is convergent, we deduce that
    \begin{equation}\label{est410}
    \abs{M_{r,1}(\w_t)} \leq C t^{\mu-1}.
    \end{equation}
Next
    \begin{equation*}\label{est411}
    \begin{split}
    M_{r,2}(\w_t) &= \a \int_r^\8 s^{-1-1/\mu} \exp\left(-\frac{it}{\pi} s^{-1/\mu}\right)\, ds\cr
    &= \frac{\a\pi\mu}{it} \int_r^\8 \left(\exp\left(-\frac{it}{\pi} s^{-1/\mu}\right)\right)'\, ds\cr
    &= \frac{\a\pi\mu}{it} \left(1-\exp\left(-\frac{it}{\pi} r^{-1/\mu}\right)\right)
    \end{split}
    \end{equation*}
and
    \begin{equation}\label{est412}
    \abs{M_{r,2}(\w_t)} \leq C t^{-1}.
    \end{equation}
For $n\geq 3$ we integrate by parts getting
    \begin{equation*}\label{est413}
    \begin{split}
    M_{r,n}(\w_t) &= \a \int_r^\8 s^{-n+1-1/\mu} \exp\left(-\frac{itn}{2\pi} s^{-1/\mu}\right)\, ds\cr
    &=\frac{2\pi\mu \a}{itn} \int_r^\8  s^{-n+2} \, \left(\exp\left(-\frac{itn}{2\pi} s^{-1/\mu}\right)\right)'\, ds\cr
    &=\frac{2\pi \mu \a}{itn} \left[-r^{-n+2} \exp\left(-\frac{itn}{2\pi} r^{-1/\mu}\right) +
    (n-2)  \int_r^\8 s^{-n+1} \, \exp\left(-\frac{itn}{2\pi} s^{-1/\mu}\right)\, ds\right].
    \end{split}
    \end{equation*}
Estimating similarly as in \eqref{est48},
    \begin{equation}\label{est414}
    \abs{M_{r,n}(\w_t)} \leq C \frac{r^{2-n}}{tn}.
    \end{equation}
Summarizing, plugging the estimates \eqref{est48}--\eqref{est414} into \eqref{warunek}, for $t>1$ we get
    \begin{equation}\label{aaarrr}
    |A_r(\w_t-\w_\8)| \leq  C t^{2\mu-2} (1+r^2),
    \end{equation}
which, together with \eqref{radial}, implies \eqref{est_prop_infty}.

\hfill $\Box$
\end{proof}

\begin{prop}\label{prop_zero}
Assume that $r>0$ is fixed. Let $\w_t$ be the vorticity of the Kaden spiral defined in \eqref{RofGamma} with $\mu\in(1/2,1)$
and $\w_0$ be given by \eqref{w_0} with $c=1$, $\a$ related to $\mu$ via \eqref{amu}. Next, let $u(t)$ and $u_0$ be divergence-free velocity fields associated with $\w_t$ and $\w_0$, respectively, by the
Biot-Savart operator \eqref{Biot_Savart}. Then
	\begin{equation}\label{zero_conv}
	\lim_{t\rightarrow 0^+}\int_{B(0,r)}|u(t)-u_0|^2dx=0.
	\end{equation}
In particular,	
    \begin{equation}\label{est_prop_zero}
    \lim_{t\to 0^+} E_r(\w_t) = E_r(\w_0).
    \end{equation}
\end{prop}

\begin{proof}
Fix $\e>0$. By \eqref{radial} and Corollary \ref{3.8} (a signed measures version of Theorem \ref{averages}) it suffices to show that
      \begin{equation}\label{est_dwie_sumy}
    \abs{\sum_{n=0}^\8 A_n + \sum_{k=1}^\8 B_k} := \abs{\sum_{n=0}^\8 r^{-2n-2} \abs{m_{r,n}(\w_t-\w_0)}^2 + \sum_{k=1}^\8 r^{2k-2} \abs{M_{r,k}(\w_t-\w_0)}^2} <\e
    \end{equation}
for $t>0$ small enough. We observe first that $A_0=m_{r,0}(\w_t-\w_0)=0$.

Since
\[
|m_{r,n}(\w_t-\w_0)|=|m_{r,n}(\w_t)-m_{r,n}(\w_0)|\leq |m_{r,n}(\w_t)|+|m_{r,n}(\w_0)|
\]
as well as
\[
|M_{r,k}(\w_t-\w_0)|=|M_{r,k}(\w_t)-M_{r,k}(\w_0)|\leq |M_{r,k}(\w_t)|+|M_{r,k}(\w_0)|,
\]
and due to the fact that both $\w_t$ and $\w_0$ satisfy assumptions of Theorem \ref{averages2},
in view of \eqref{est_m} and \eqref{est_M}, we find $N\in \mathbb N$ such that
\begin{equation}\label{prosiur1}
    \sum_{n=N}^\8 \abs{A_n} + \sum_{k=N}^\8 \abs{B_k} <\e/3.
\end{equation}
Assume that $k=1,...,N-1$. Using Lemma \ref{density} and the estimate $\abs{e^{i\tau}-1} \leq C |\tau|$  for $\tau \in \mathbb R$ we get
    \begin{equation*}
    \begin{split}
    \abs{M_{r,k}(\w_t-\w_0)} &\leq \a \int_r^\8 s^{-k+1-1/\mu} \left|\exp\left(-\frac{itk}{2\pi} s^{-1/\mu}\right)-1\right|\, ds\cr
    &\leq C tk\int_r^\8 s^{-k+1-2/\mu}\, ds\cr
    &\leq C t r^{-k+2-2/\mu}.
    \end{split}
    \end{equation*}
Consequently,
\begin{equation}\label{prosiur}
    \sum_{k=1}^{N-1} \abs{B_k} \leq C t^2 \sum_{k=1}^{N-1} r^{2k-2} r^{-2k+4-4/\mu}\leq C t^2 N r^{2-4/\mu}<\e/3,
\end{equation}
where the last inequality holds for $t>0$ small enough.

We now turn to estimating $A_n$ for $n=2,...,N-1$ (the case $n=1$ is a bit more delicate and will be treated separately). By using Lemma \ref{density} one more time,
    \begin{eqnarray}\label{prosiur2}
    \abs{m_{r,n}(\w_t-\w_0)} &\leq \a \int_0^r s^{n+1-1/\mu} \left|\exp\left(\frac{itn}{2\pi} s^{-1/\mu}\right)-1\right|\, ds\nonumber\\
    &\leq C t n \int_0^r s^{n+1-2/\mu}\, ds\nonumber\\
    &\leq C t r^{n+2-2/\mu}.
    \end{eqnarray}
We have used $n+2-2/\mu>0$. The latter is a consequence of  $n\geq 2$ and $\mu \in (1/2,1)$. Notice that similar estimate holds for $n=1$ and $\mu \in (2/3,1)$. Indeed, then $2-2/\mu>-1$ and so
\[
\abs{m_{r,1}(\w_t-\w_0)}\leq \a \int_0^r s^{2-1/\mu} \left|\exp\left(\frac{it}{2\pi} s^{-1/\mu}\right)-1\right|\, ds
\leq C t\int_0^r s^{2-2/\mu}\, ds\leq C t r^{3-2/\mu}.
\]
Consequently, see also \eqref{prosiur2}, for $\mu \in (2/3,1)$ and $t>0$ small enough
\[
\sum_{n=1}^{N-1}|A_n|\leq Ct^2r^{2-4/\mu}<\e/3,
\]
which together with \eqref{prosiur1} and \eqref{prosiur} gives the claim of Proposition \ref{prop_zero} in the range of parameters $\mu\in (2/3,1)$.

In the remaining case $\mu \in(1/2,2/3]$ we still have to deal with $m_{r,1}$. Again we split our considerations into two cases, $\mu<2/3$ and $\mu=2/3$. In the first one we have
    \begin{equation*}
    \begin{split}
    \abs{m_{r,1}(\w_t-\w_0)} &\leq \a \int_0^r s^{2-1/\mu} \left|\exp\left(\frac{it}{2\pi} s^{-1/\mu}\right)-1\right|\, ds = \a\left( \int_0^{t^\mu}... + \int_{t^\mu}^r ...\right)\cr
    &\leq C \int_0^{t^\mu} s^{2-1/\mu}\, ds + C t \int_{t^\mu}^r s^{2-2/\mu}\, ds \cr
    &\leq C t^{3\mu-1} + C t \left[s^{3-2/\mu}\right]^{s=t^\mu}_{s=\8} \leq C t^{3\mu-1}.
    \end{split}
    \end{equation*}
Therefore, in the case $\mu\in (1/2,2/3)$, taking into account also  \eqref{prosiur2}, we estimate
\begin{equation}\label{last}
    \sum_{n=1}^{N-1} \abs{A_n} \leq C t^{6\mu-2}r^{-4} + C t^2 \sum_{n=2}^{N-1} r^{-2n-2} r^{2n+4-4/\mu}\leq C t^{6\mu-2} N \max(r^{-4},r^{2-4/\mu})<\e/3,
\end{equation}
if $t>0$ is small enough. Hence, the claim of Proposition \ref{prop_zero} follows for $\mu\in (1/2,2/3)$.

Similarly in the case $\mu=2/3$
\[
\abs{m_{r,1}(\w_t-\w_0)}\leq C t^{3\cdot2/3-1}+Ct \int_{t^{2/3}}^r s^{-1}\leq Ct+Ct(\ln r-2/3\ln t).
\]
Hence for small enough $t>0$ (after taking into account \eqref{prosiur2})
\begin{equation}\label{naprawdelast}
\sum_{n=1}^{N-1} \abs{A_n} \leq C(t-t\ln t)^2r^{-4}+ Ct^2r^{-4}\ln^2 r+Ct^2r^{-4}<\e/3.
\end{equation}
Estimates \eqref{last} and \eqref{naprawdelast} finish the proof of \eqref{est_dwie_sumy} for $\mu\in (1/2,2/3]$.

\hfill $\Box$
\end{proof}

At the end let us remark that Kaden's spiral is continuous in a certain sense. The proof is very similar to the proof of Proposition \ref{prop_zero}, so we only provide a sketch.

\begin{prop}\label{prop_cont}
Assume that $r>0$ is fixed. Moreover take $0<t, t_0<\infty$. Let $\w_t, \w_{t_0}$ be vorticities of Kaden spirals at times $t, t_0$ defined in \eqref{RofGamma} with $\mu\in(1/2,1)$. Next, let $u(t)$ and $u(t_0)$ be the divergence-free velocity fields associated to $\w_t$ and $\w_{t_0}$, respectively, by the
Biot-Savart operator \eqref{Biot_Savart}. Then
	\begin{equation}\label{cont}
	\lim_{t\rightarrow t_0}\int_{B(0,r)}|u(t)-u(t_0)|^2dx=0.
	\end{equation}
\end{prop}

\begin{proof}
The beginning of the proof goes the same way as in Proposition \ref{prop_zero}.
Let $E_n$ and $D_k$ be defined as
\[
E_n:=r^{-2n-2} \abs{m_{r,n}(\w_t-\w_{t_0})}^2, D_k := r^{2k-2} \abs{M_{r,k}(\w_t-\w_{t_0})}^2,
\]
then,
due to the fact that $\w_t$ and $\w_{t_0}$ (in view of \eqref{omega-ball-2}) satisfy assumptions of Theorem \ref{averages2}, estimates \eqref{est_m},\eqref{est_M} guarantee that
for any $\e>0$ we can find $N>0$ such that
\begin{equation}\label{est_dwie_sumy_cont}
    \sum_{n=N}^\8 |E_n| + \sum_{k=N}^\8 |D_k| <\e/3.
    \end{equation}
To show \eqref{cont} and thus finish the proof of Proposition \ref{prop_cont} we only need to
estimate $E_n$ for $n=0,1,...,N-1$ and $D_k$ for $k=1,...,N-1$. On the one hand, in view of \eqref{omega-ball-2},
$E_0=0$. Next, utilizing Lemma \ref{density}, we notice that
\[
\abs{m_{r,n}(\w_t-\w_{t_0})} \leq \a \int_0^r s^{n+1-1/\mu} \left|\exp\left(\frac{i(t-t_0)n}{2\pi} s^{-1/\mu}\right)-1\right|\, ds,
\]
\[
\abs{M_{r,k}(\w_t-\w_{t_0})} \leq \a \int_r^\8 s^{-k+1-1/\mu} \left|\exp\left(-\frac{i(t-t_0)k}{2\pi} s^{-1/\mu}\right)-1\right|\, ds,
\]
so that for $t\rightarrow t_0$ the required estimates of the initial $E_n, D_k$ may be achieved exactly the same way as in the proof of Proposition \ref{prop_zero}.

\hfill $\Box$
\end{proof}

At the end let us provide a simple corollary.
\begin{cor}\label{contin}
Let $\w_t$ be the Kaden spiral. For any fixed $r>0$, the function $t\rightarrow E_r(\w_t)$ is continuous.
\end{cor}

\medskip
\noindent
{\bf Acknowledgements.}\\
T. Cie\'slak was partially supported by the National Science Centre (NCN), Poland, under grant 2013/09/D/ST1/03687. K. Oleszkiewicz was partially supported  by  the National Science Centre, Poland, project number 2015/18/A/ST1/00553. M. Preisner had a post-doc at WCMCS in Warsaw, where he met Kaden's spirals and started working in a project which led to the present article. He wishes to express his gratitude for support and hospitality. M. Preisner was partially supported by National Science Centre (NCN), Poland, Grant No. 2017/25/B/ST1/00599.
M.~Szuma\'nska was partially supported  by National Science Centre (NCN), Poland, Grant No.\ 2013/10/M/ST1/00416 \emph{Geometric curvature energies for subsets of the Euclidean space.} The second and fourth authors include double affiliations since a part ot the work have been conducted during their leave from University of Warsaw to IMPAN in the academic year 2016/17. The authors declare no conflict of interest.

\end{document}